\documentclass[11pt]{article}

\usepackage[margin=1in]{geometry}
\usepackage[utf8]{inputenc} 
\usepackage[T1]{fontenc}    
\usepackage{lmodern}
\usepackage{booktabs}       
\usepackage{amsfonts}       
\usepackage{nicefrac}       
\usepackage{microtype}      
\usepackage{xcolor}         

\usepackage[round]{natbib}

\usepackage{mathtools}

\usepackage{amsmath,amsthm,amsfonts,amssymb,mathdots,array,mathrsfs,bm,bbm,stmaryrd,graphicx,xcolor}
\usepackage{algorithm}
\usepackage{breakcites}
\colorlet{shadecolor}{gray!80}

\usepackage{enumerate}
\usepackage{inputenc}

\usepackage{graphicx} 
\usepackage{booktabs,balance}
\usepackage{rotating}
\usepackage{boldline}
\usepackage{makecell}
\usepackage{multirow}
\usepackage{balance}

\usepackage[colorlinks,linkcolor=red,filecolor=blue,citecolor=blue,urlcolor=blue]{hyperref}
\usepackage[capitalize,noabbrev]{cleveref}

\usepackage{url}
\theoremstyle{plain}
\newtheorem{theorem}{Theorem}[section]

\newtheorem{lemma}[theorem]{Lemma}

\theoremstyle{definition}
\newtheorem{definition}[theorem]{Definition}

\usepackage{enumitem}
\usepackage{nicefrac}
\usepackage{url}
\usepackage{wrapfig}
\usepackage[utf8]{inputenc}
\usepackage{graphicx}
\usepackage{latexsym}
\usepackage{amsmath}
\usepackage{url}
\usepackage{float}
\usepackage{tikz}
\usepackage{fancyvrb}
\usepackage{listings}
\usepackage{xcolor}
\usepackage{color}
\usepackage{soul}
\usepackage{multirow}
\usepackage{multicol}
\usepackage{subfig}
\usepackage{booktabs}
\usepackage{caption}

\usepackage[toc,page,header]{appendix}
\usepackage{minitoc}
\usepackage{todonotes}

\usepackage{xcolor}

\usepackage{mathtools}
\usepackage{algorithm}
\usepackage{algorithmicx}
\usepackage{algpseudocode}
\usepackage{subfig}

\algdef{SE}[DOWHILE]{Do}{doWhile}{\algorithmicdo}[1]{\algorithmicwhile\ #1}%

\DeclareMathOperator{\argmax}{argmax}
\DeclareMathOperator{\argmin}{argmin}

\usepackage[capitalize]{cleveref}

\begin{document}

\title{\bf \huge Faster  Algorithms for Generalized Mean Densest Subgraph Problem}

\author{\vspace{0.3in}\\\textbf{Chenglin Fan, Ping Li and Hanyu Peng} \\\\
Cognitive Computing Lab\\
Baidu Research\\
No.10 Xibeiwang East Road, Beijing 100193, China\\
10900 NE 8th St. Bellevue, Washington 98004, USA\\
  \texttt{\{chenglinfan2020,\ pingli98,\  hanyu.peng0510\}@gmail.com}
}
\date{\vspace{0.5in}}

\maketitle

\begin{abstract}\vspace{0.2in}

\noindent\footnote{This work was initially  submitted in February  2022 to the SIGKDD'22 conference. The authors sincerely thank the helpful comments from the reviewers of past submissions, e.g., the Program Committee of WWW'23.}The densest subgraph of a large graph usually refers to some subgraph with the highest average degree, which has  been extended
to the family of $p$-means dense subgraph objectives by~\citet{veldt2021generalized}. The $p$-mean densest subgraph problem seeks a subgraph with the highest average $p$-th-power degree, whereas the standard densest subgraph problem seeks a subgraph with a simple highest average  degree.
It was shown  that the standard peeling algorithm can perform arbitrarily poorly on generalized objective when $p>1$ but uncertain when $0<p<1$. In this paper, we are the first to show that a standard peeling algorithm can still  yield $2^{1/p}$-approximation for the case $0<p < 1$.
\citet{veldt2021generalized} proposed a new generalized peeling algorithm (GENPEEL), which for $p \geq 1$ has an approximation guarantee ratio
$(p+1)^{1/p}$, and time complexity $O(mn)$, where $m$ and $n$ denote the number of edges and nodes in graph respectively.
In terms of algorithmic contributions, we propose a new and faster
generalized peeling algorithm (called \textsc{GENPEEL++} in this paper), which for $p \in [1, +\infty)$ has an approximation guarantee ratio
$(2(p+1))^{1/p}$, and time complexity $O(m(\log n))$, where $m$ and $n$ denote the number of edges and nodes in graph, respectively.
This approximation ratio  converges to 1 as $p \rightarrow \infty$.

\vspace{0.2in}

\noindent Our experiments  show that  GENPEEL++  can obtain extremely close approximations to the previous GENPEEL algorithm, and it performs significantly (e.g., up to 10x) faster than the GENPEEL algorithm in several large-scale real-world datasets for both $p<1$ and $p>1$. \vspace{0.05in}
\end{abstract}

\newpage

\section{Introduction}

The problem of dense components detection in a graph has been
extensively studied~\citep{lee2010survey,gionis2015dense,tsourakakis2015k,zhang2017hidden, sariyuce2018peeling,shin2018patterns,  ma2020efficient, liu2021efficient}.
Various definitions of density have been explored.The problem of finding dense subgraphs can be considered as an variant of correlation clustering~\citep{becker2005survey}. The difference between them is that the former only cares about the internal edges of a subgraph, while the latter counts both the
internal edges and external edges connecting to the rest of the graph.
In this paper, the notation of density we are interested in is, roughly speaking, the mean average degree of a subgraph. More details about the definitions of the mean densest subgraph are given in Section~\ref{sec:pre}.

From a theoretical perspective, dense components in graph have many interesting properties. For example, dense components have small diameters (the shortest path between two nodes). Also, dense components are robust, so the nodes in components may still be connected after part of the edges/connections are broken.
Based on those properties, dense components have been identified in  enhanced understanding of various types of networks in the real world.
Among the best-knowns are communities in social networks~\citep{sozio2010community}, DNA motifs~\citep{fratkin2006motifcut}, trending topics in social media~\citep{angel2012dense}, brain networks~\citep{lanciano2020explainable}, the World Wide Web~\citep{kumar1999trawling}, and financial markets~\citep{nagurney2003innovations}, etc.
Although there are exponentially many subgraphs, the problem of finding the densest subgraph of a given graph  can be solved optimally
in polynomial time~\citep{goldberg1984finding}. In addition, \citet{charikar2000greedy} showed that we can find a 2 -approximation solution
to the densest subgraph problem in linear time using a very simple greedy algorithm, which is similar to a greedy strategy  previously studied by~\citet{asahiro1996greedily}.	
Hardness results  and positive solution for dense subgraph discovery have been studied extensively in theory~\citep{karp1972reducibility,charikar2000greedy,andersen2009finding, lee2010survey,pattillo2013maximum}.  The size
of the graph involved could be very large, so having a fast algorithm for finding
an approximately dense subgraph is extremely useful.
The problem of finding the densest subgraph in a graph without size restriction can be solved in polynomial time.
However when there is a size
constraint specified, namely finding the densest subgraph of exactly $k$ vertices (DkS), the
densest $k$-subgraph problem becomes NP-hard~\citep{feige2001thedense,asahiro2002complexity}. The Generalized objective functions for dense subgraph was first introduced in~\citet{veldt2021generalized}.  It was shown that standard densest subgraph problem is obtained as special case when  $p = 1$.
In a word, the new generalized density objective functions unify a number of previous definitions.

It was shown in~\citet{veldt2021generalized} that  objective function is polynomial-time solvable when $p \geq 1$, by  repeatedly calling submodular minimization.  The computational complexity for case  $p<1$ is still unknown.
The standard simple peeling algorithm (SIMPEEL) to obtain approximation is  repeatedly removing a single vertex at a time in order to shrink a graph down into a denser subgraph.
\citet{charikar2000greedy} showed that yields 2-approximation by  iteratively removing the node with smallest degree in current graph. However, despite seeming like a natural approach, that well-known standard peeling algorithm, which  provides a $2$-approximation for the $p = 1$ objective~\citep{charikar2000greedy,khuller2009finding}, can output arbitrarily bad results when $p > 1$.
In this paper, we show that SIMPEEL still can yield $2^{1/p}$-approximation for the case $0<p < 1$.
In order to solve the case for $p>1$, a more sophisticated but still fast peeling algorithm (GENPEEL) with a $(1+p)^{1/p}$ approximation guarantee when $p\geq 1$ was proposed by~\citet{veldt2021generalized}. They also presented the fact  that GENPEEL can outperform SIMPEEL in finding dense subgraphs.  For example, on many real-world graphs, they found that running GENPEEL with a value of $p$ slightly larger than 1 will typically produce sets with a better average degree than SIMPEEL.

\newpage

The main difference between SIMPEEL and GENPEEL is the following. SIMPEEL only considers the affection of its own degree when removing a node.   Concerning the contribution to the generalized objective, GENPEEL  considers  the  degree of both itself  and  its  neighbors in the graph.
That kind of   ``foresight'' when removing nodes,  is not presented in the strategy of the simple peeling algorithm.
However, that ``foresight'' needs extra time to obtain and update, as the removal of a vertex  $v$ not only affects the ``foresight'' of its neighbors, but also its second neighbors (the vertex set with two hops to $v$).  Hence GENPEEL is a bit slower than SIMPEEL, and it has $O(mn)$  time complexity since ``foresight'' was update when removing each node.
In this paper, we found that ``foresight'' is unnecessary to update
in each step,  only roughly $O(\log n)$ times enough to obtain
a good approximation. Based on that observation, we propose the ``GENPEEL++'' for $p\geq 1$ in this paper,  which has only $O(m\log n)$ time complexity.

We compare  our proposed GENPEEL++ with  the previous GENPEEL~\citep{veldt2021generalized} on a range of different sized graphs from various domains, including social networks, road networks, citation networks, and web networks. We also show that SIMPEEL yields constant approximation  when $0<p<1$. Equipped with close approximation to GENPEEL, our proposed GENPEEL++ runs much faster than GENPEEL in large datasets, also uncovers different meaningful notions of dense subgraphs varying parameter $p$.

\vspace{0.1in}
In summary, in this paper we present the following \textbf{contributions}:
\begin{itemize}
	\item We revisit the generalized mean densest subgraph problem and show that \textsc{SIMPEEL} still can yield $2^{1/p}$-approximation for case $0<p < 1$. Hence, for $p<1$ we theoretically resolve the open problem raised in~\citet{veldt2021generalized}.

	\item we propose a faster
generalized peeling algorithm (named \textsc{GENPEEL++} in this paper), which for $p \in [1, +\infty)$ has an approximation guarantee ratio
$(\frac{p+1}{1-c})^{1/p}$, and time complexity  $O(m(\log n) / (\log \frac{1}{1-c})$  for any constant $c$ ($0<c<1$). It improves the previous $O(mn)$ time algorithm, where $m$ and $n$ denote the number of edges and nodes in graph respectively.
That  converges to 1 as $p \rightarrow \infty$.

	\item In the experiments, we  compare the performance of different peeling algorithms including \textsc{SIMPEEL}, \textsc{GENPEEL}, \textsc{GENPEEL++}.
	For thorough comparison, we conducted experiments on all datasets reported in~\citet{veldt2021generalized}.  The improvements are consistent for both $p<1$ and $p>1$
	(Note that the experiment  in~\citet{veldt2021generalized}  runs \textsc{GENPEEL}
	for both  $p<1$ and  $p>1$).
	It is shown in the experiments, the approximation errors of \textsc{GENPEEL++}  do not differ much compared  to \textsc{GENPEEL},
	but \textsc{GENPEEL++}  performs significantly faster (10x times faster)   in large real datasets,  coming from numerous domains.
\end{itemize}
%

\section{Technical Preliminaries\label{sec:pre}}

We follow  some notions  in~\citet{veldt2021generalized}. For a given undirected graph $G = (V,E)$ and each vertex $v\in V$, let $N(v) = \{u \in V \colon (u,v) \in E\}$ denote the neighborhood of node $v$, and $d_v = |N(v)|$ be its degree. Note that $v\notin N(v)$.
For an  arbitrary  set $S \subset V$, let $E(S)$ denote the set of edges between all pairs of  nodes in $S$ and $d_v(S) = |N(v) \cap S|$ denote  the degree of $v$ in the subgraph induced by $S$. Hence, we have   $d_v(S) = 0$ if $v\notin S$.

\paragraph{Dense Subgraph Problems}

The densest subgraph problem seeks a subgraph $S\subseteq V$ maximizing the density $f(S)$, defined as the ratio between the number of edges and nodes:
\begin{equation}
\max_{S \subseteq V} f(S)= \frac{2|E(S)|}{|S|}=\frac{\sum_{v \in S} d_v(S)}{|S|}.
\end{equation}
We just let $f(S)=0$ when $|S|=0$.
This problem is known to have a polynomial-time solution~\citep{goldberg1984finding}, as well as a fast greedy peeling algorithm that is guaranteed to return a $2$-approximation~\citep{charikar2000greedy,khuller2009finding}.

 A new generalized  dense subgraph objective based on generalized means of degree sequences was introduced by~\citet{veldt2021generalized}.
They extended  the density function to $p$-th-power  as
\begin{equation}
\label{fp}
{f_p(S) = \sum_{v \in S} \frac{d_v(S)^p}{|S|}.}
\end{equation}
The $p$-density of $S \subseteq V$ is
\begin{equation}
\label{subgraphpmean}
{ M_p(S) = f_p^{1/p}(S),p\neq 0.}
\end{equation}
The \emph{$p$-mean  densest subgraph problem} is then to find a set of nodes $S$ that maximizes $M_p(S)$.
For finite $p > 0$, maximizing $M_p(S)$ is equivalent to maximizing $[M_p(S)]^p$.

\section{Structure Property of the Optimal Solution}

Consider a given graph $G(V,E)$, and some   subgraph  $T$ that maximizes $M_p(S)$.
A natural question is that: can the degree distribution of vertices in $T$  be arbitrary? In this paper, our investigation  begins with a simple case when $p=1$.

\begin{lemma}\label{lem:p=1}
For a given graph $G(V,E)$ and   a densest subgraph   $T=\arg\max_{S \subseteq V} f(S)$, we have
 $$d_u(T) \geq   f(T)/2, \forall u \in T$$
\end{lemma}
\begin{proof}
The cases with  $|T|=0$ or $1$ are obvious. For $|T|>1$, we proceed  the proof by contradiction. Suppose some node $u$ in $T$ satisfies that
$$d_u(T) <   f(T)/2$$
We have $$f(T\setminus\{u\})=\frac{f(T)|T|-2d_u(T)}{|T|-1} > \frac{f(T)|T|-f(T)}{|T|-1}.$$
Then we have
$$f(T\setminus\{u\}) > f(T).$$
This means $T\setminus\{u\}$ is a  more dense subgraph instead of $T$.
That is a contradiction and completes the proof.
\end{proof}

\begin{definition} [\citet{veldt2021generalized}]
For a graph $G(V,E)$, some node  subset $S\subseteq V$, and arbitrary node $v \in S$, the   loss/decrease value of object $f_p(S)$ resulting from removing $v$ is defined as follow:
\begin{equation}
\label{Delta}
\Delta_v(S) = d_v(S)^p + \sum_{u \in N(v) \cap S} d_{u}(S)^p - [d_u(S) - 1]^p.
\end{equation}
\end{definition}
We now consider the property of
 the loss value.
\begin{lemma} \label{lem:p>0}
For a given graph $G(V,E)$ and $p>0$, and    subgraph   $T=\arg\max_{S\subseteq V} f(S)$, we have
 $$\Delta_v(T) \geq   f_p(T), \forall v \in T$$
\end{lemma}
\begin{proof}
The cases when  $|T|=0$ or $1$ are obvious. We  consider the case $|T|>1$ below.

Proof by contradiction. Suppose some node $v$ in $T$ satisfies
$\Delta_v(T) <  f_p(T)$, we have $$f_p(T\setminus\{v\})=\frac{f_p(T)|T|-\Delta_v(T)}{|T|-1} > \frac{f_p(T)|T|-f_p(T)}{|T|-1}.$$
Then we have
$$f_p(T\setminus\{v\}) > f_p(T).$$
This means $T\setminus\{v\}$ is a  more dense subgraph instead of $T$.
The proof completes by contradiction.
\end{proof}

 \newtheorem{fact}[theorem]{Fact}
Below, we consider structure property of the  densest subgraph  $T=\argmax_{S\subseteq V} f_p(S)$.

\begin{fact}
\label{lem:property}
For a given graph $G(V,E)$  and some   specific subgraph  $S\subseteq V $ $(|S|\geq 1)$, and $f_p(S)=Z^p$ for some positive  $Z$.
We have $Z^p=f_p(S)\leq  f(S)$ when $0<p\leq 1$, since function $f(x)=x^p$ is monotone decreasing when $0<p\leq 1$.
\end{fact}

\begin{lemma}\label{lem:delta}
For a given graph $G(V,E)$  and  a densest  subgraph   $T=\arg\max_{S\subseteq V} f_p(S)$. Let $v=\arg\min_{u\in T} d_u(T)$,  we have
   $\Delta_v(T) \leq  2(d_v(T))^p$ when $0<p\leq 1$.
\end{lemma}
\begin{proof}
The cases when  $|T|=0$ or $1$ is easy. We  consider the case $|T|>1$ below.
 $\Delta_v(T) = d_v(T)^p + \sum_{u \in N(v) \cap T} (d_{u}(T)^p - [d_u(T) - 1]^p).$ Which is
 \begin{align}\notag
 &=d_v(T)^p \\\notag
 &+ \sum_{u \in N(v) \cap T} \frac{(d_{u}(T)^p - [d_u(T) - 1]^p)(d_{u}(T)^{1-p} + [d_u(T) - 1]^{1-p})}{d_{u}(T)^{1-p} + [d_u(T) - 1]^{1-p}}\\ \notag
 &=d_v(T)^p +\\ \notag
 &\sum_{u \in N(v) \cap T} \frac{1+d_u(T)^p(d_u(T)-1)^{1-p}-(d_u(T)-1)^{p}d_u(T)^{1-p}}{d_{u}(T)^{1-p} + [d_u(T) - 1]^{1-p}}\\ \notag
 &\leq
 d_v(T)^p + \sum_{u \in N(v) \cap T} \frac{1}{d_{u}(T)^{1-p} + [d_u(T) - 1]^{1-p}}\\\notag
  &<
 d_v(T)^p + \sum_{u \in N(v) \cap T} \frac{1}{d_{u}(T)^{1-p}}\\\notag
 &\leq
 d_v(T)^p + \sum_{u \in N(v) \cap T} \frac{1}{d_{v}(T)^{1-p}}\\\notag
&=
 d_v(T)^p + d_v(T)^{p}=2d_v(T)^{p}
 \end{align}
\end{proof}

\noindent Based on the Lemma above, we show the bound between  $d_v(T)$~and~$f_p(S)$.
\begin{lemma}\label{lem:bound}
For a given graph $G(V,E)$ and $0<p\leq 1$, and    subgraph  $ T=\argmax_{S\subseteq V}f_p(S)$. We have
 $$d_v(T)^p \geq   \frac{f_p(T)}{ 2}, \forall v \in T$$
\end{lemma}
\begin{proof}
Proof by contradiction.
Let $f_p(T)=Z^p$.
Suppose   $u=\arg\min_{v\in T} d_v(T)$ in $T$ satisfies that
$$d_u(T)^p < \frac{Z^p}{2}  $$
We have $$f_p(T\setminus\{u\})=(|T|Z^{p}-\Delta_u(T))/(|T|-1)$$
Based on Lemma~\ref{lem:delta}, we have
\begin{align*}
   |T|Z^{p}-\Delta_u(T)& \geq |T|Z^{p}-2d_u(T)^p\\
   &\geq (|T|-1)Z^{p}+ Z^p-2d_u(T)^p\\
   &>(|T|-1)Z^{p}
\end{align*}
Then we have
$$f_p(T\setminus\{u\}) > f_p(T).$$
That is a contradiction. The Lemma is proved.
\end{proof}

 Based on Lemma above, we find that $d_v(T)^p \geq   \frac{f_p(S)}{ 2}, \forall v \in T   $, where $ T=\argmax_{S\subseteq V}f_p(S)$. This can be used to guide the design of approximation of the case $0<p \leq 1$.

\section{Algorithms}

We now present our proposed faster algorithm \textsc{GENPEEL++} for the $p$-mean densest subgraph problem.  For $p>1$, it was shown that  \textsc{SIMPEEL}~\citep{asahiro2000greedily,charikar2000greedy,khuller2009finding} can return arbitrarily bad results by~\citet{veldt2021generalized}.
Perhaps surprisingly, in this paper we show that \textsc{SIMPEEL} algorithm  can still yield constant approximation for case $0<p<1$.

\subsection{Success  of the Standard Peeling Algorithm for $0<p<1$}

The standard peeling algorithm for  densest subgraph problem  starts with the entire graph $G(V,E)$ and repeatedly removes the minimum degree node until no more node remains. We  refer to this algorithm generically as \textsc{SIMPEEL}, which produces a set of $n$ nested subgraphs $S_1 \supset  S_2 \supset \hdots, \supset S_n$, one of which is guaranteed to provide at least a $2$-approximation to the standard densest subgraph problem~\citep{charikar2000greedy}.

Given the success of this procedure for $p=1$, it is natural to wonder whether it can be used to obtain optimal or near optimal solutions for other values like $p<1$. Perhaps surprisingly, we are able to show that the simple peeling algorithm  performs not too bad  for any $0<p < 1$.

\begin{lemma}
	\label{lem:success}
	Let $0<p < 1$  be  a fixed constant. Applying \textsc{SIMPEEL} on  $G(V,E)$ will yield a $2$-approximation  for the $p$-th-power degree objective and a $2^{1/p}$-approximation  for the p-density objective.
\end{lemma}
\begin{proof}

Consider a given graph $G(V,E)$ and $0<p\leq 1$, and some   subgraph  $$T=\arg\max_{S\subseteq V} f_p(S).$$
Define  $f_p(T)=Z^p$.
	Let $S$ denote the set maintained by the greedy algorithm right before the first node $v \in T$ is removed by \textsc{SIMPEEL}. Since $v$ is the first node to be removed by the peeling algorithm, we know that $S \supseteq T$ and $d_v(T) \leq d_v(S)$. Because $v = \argmin_{u \in S}d_u(S)$,
	we have $$d_u(S)\geq d_v(S) \geq d_v(T)$$
	Now let us compute $f_p(S)$, which is
		\begin{align*}
	f_p(S)=&\frac{\sum_{u \in S} d_u(S)^p}{|S|}\\
	&\geq  d_v(S)^p\\
	& \geq d_v(T)^p
	\end{align*}
Define $f_p(T)=Z^p$, we have  $d_v(T)^p\geq Z^p/2$ based on Lemma~\ref{lem:bound}. Then we have
\begin{align*}
	f_p(S)\geq d_v(T)^p \geq Z^p/2,
	\end{align*}
	and hence we have  $\frac{f_p(T)}{f_p(S)}\leq 2$.
	Since
\begin{equation*}
 M_p(S)=(f_p(S))^{1/p},
\end{equation*}
we know that the approximation ratio is $\frac{M_p(T)}{M_p(S)}= (\frac{f_p(T)}{f_p(S)})^{1/p}\leq (2)^{1/p}=2^{1/p}$.
\end{proof}

\subsection{Faster Peeling Algorithm when $p \geq 1$}

Below we provide the algorithm workflow of the standard simple peeling algorithm.

\begin{algorithm}[H]
	\caption{Standard Simple Peeling Algorithm (\textsc{SIMPEEL})~\citep{charikar2000greedy}}
	\label{alg:SIMPEEL}
	\begin{algorithmic}
		\State \textbf{Input}: $G = (V,E)$, parameter $0<p \leq  1$
		\State \textbf{Output}: Set $S' \subseteq V$.
		\State $S_0 := V$
		\For{$i := 1$ to $n$}
		
		\State $\ell := \argmin_v d_v(S_{i-1})$
		\State $S_i := S_{i-1}\backslash\{\ell\}$
		
		\EndFor
		\State Return $S':=\max_i f_p(S_i)$.
	\end{algorithmic}
\end{algorithm}
In order to make the comparison between Simple Peeling Algorithm and Generalized Peeling Algorithm more clear, we include Algorithm~\ref{alg:genpeel}~\citep{veldt2021generalized} here for completeness.

\begin{algorithm}[H]
	\caption{Generalized Peeling Algorithm (\textsc{GenPeel})~\citep{veldt2021generalized}}
	\label{alg:genpeel}
	\begin{algorithmic}
		\State \textbf{Input}: $G = (V,E)$, parameter $p \geq 1$
		\State \textbf{Output}: Set $S' \subseteq V$, satisfying $f_p(S') \geq \frac{1}{p+1} \max_S f_p(S)$.
		\State $S_0 := V$
		\For{$i := 1$ to $n$}
		\State $\ell := \argmin_j \Delta_j(S_{i-1})$
		\State $S_i := S_{i-1}\backslash\{\ell\}$
		\State Update $\Delta_j(S_{i})$
		\EndFor
		\State Return $\max_i f_p(S_i)$.
	\end{algorithmic}
\end{algorithm}

Algorithm~\ref{alg:SIMPEEL} could not be extended to case  $p > 1$,
because the property in Lemma~\ref{lem:bound} could not be extended to the case $p>1$. In order to solve the case for $p>1$,  a more sophisticated but still fast peeling algorithm (GENPEEL) with a $(1+p)^{1/p}$ approximation guarantee when $p\geq 1$ was proposed by~\citet{veldt2021generalized}.

In  each step of GENPEEL, one pick the vertex $\ell=\arg\min_{j\in S}\Delta_j(S) $ where $S$ is the current graph remains. Also  the value of  $\Delta_j(S) $ is dynamically updated in each step of
GENPEEL.

\begin{algorithm}[t]
	\caption{Faster Generalized Peeling Algorithm (\textsc{GENPEEL++})}
	\label{alg:GENPEEL++}
	\begin{algorithmic}
		\State \textbf{Input}: $G = (V,E)$, parameter $p \geq 1$
		\State \textbf{Output}: Set $S' \subseteq V$, satisfying $f_p(S') \geq \frac{1}{2(p+1)} \max_{T\subseteq V} f_p(T)$.
		\State $S:= V$, $I :=0,  S_0:=V $\;
		
		\While{$|S|>0$}
		
		\State Compute the $\Delta_j(S)$ for the  vertices in $S$\;
		\State  Sort the current vertex set in $S$ based on $\Delta_j(S)$ as $v_1,v_2,...,v_{|S|}$ from small to large\;

		\For{$i := 1$ to $|S|/2$}
		\State $S_{I+i} := S_{I+i-1}\backslash\{v_i\}$\;
		
		\EndFor
			\State $I:=I+|S|/2$\;
		\State $S:=S\setminus \bigcup_{1 \leq i\leq |S/2|} \{v_i\}$ \;
	
		\EndWhile
		\State Return $\max_I f_p(S_I)$ when $I:=1,2,3,... n$.
	\end{algorithmic}
\end{algorithm}

We put the  Algorithm~\ref{alg:genpeel}  proposed by~\citet{veldt2021generalized} here just for comparison. As discussed in the introduction, the Algorithm~\ref{alg:genpeel} captures the ``insight'', while the only disadvantage is its efficiency.
In this paper, we found that value set of $\Delta_j(S)$ was not necessary to update
in each step,  only roughly $O(\log n)$ times  enough to obtain
a constant approximation. Based on that observation, we propose the ``GENPEEL++'' for $p\geq 1$ in this paper,  which has only $O(m\log n)$ time complexity. More details  is given in
Algorithm~\ref{alg:GENPEEL++}.
The  Algorithm~\ref{alg:GENPEEL++} repeat steps below until at most one node left:
 (i) Compute (Recompute) the $\Delta_j(S)$ for the  all the vertices of current graph $S$. (ii) Sort the current vertex set in $S$ based on $\Delta_j(S)$ as $v_1,v_2,...,v_i,...,$\; (iii) Remove the first half $(1/2)$ of vertices of $S$ based on that order.  Note that ``half'' can be replaced by any constant $c$ where $0 <c< 1$. The number of iterations is bounded by $O(\log n)$.
Once the nested subgraphs are obtained, we try those
subgraphs to find a good approximation solution like previous reference.

\begin{theorem}\label{thm:result}
	Let $G = (V,E)$ be a graph, $p \geq 1$, and $T$ be the $p$-mean densest subgraph of $G$,  \textsc{GenPeel++} returns a subgraph ${S}$ satisfying $2(p+1) f_p({S})  \geq f_p(T)$, i.e., $(2(p+1))^{1/p}M_p({S}) \geq  M_p(T)$ in $O(m\log n)$ time.
\end{theorem}
The proof here follows the  spirit of Theorem 4.4~\citep{veldt2021generalized}.
\begin{proof}
Let $Y = f_p(T)$, hence $\sum_{i \in T} d_i(T)^p - |T| Y  = 0$. Since $T$ is optimal, removing a node $j$ will produce a set with $p$-density at most $Y$, and therefore, we have $Y \leq \Delta_j(T)$.
	Observe that for any set $S \supseteq T$ and $j \in T$ we have $\Delta_j(T) \leq \Delta_j(S)$.

	Let $S_1\supset S_2 \supset,...,S_{I},...,$ be a set of nested subgraphs we maintain
	in Algorithm~\ref{alg:GENPEEL++}.
	Let $I$ be index such that $T\subseteq S_I$  and $T\setminus S_{I+z} \neq \emptyset$ where $z=|S_I|/2$.
	Namely,  $S_{I}$ be the minimum subgraph containing $T$. Since $j$ is one of the first ''half'' of nodes in $S_I$ to be removed, we know that $S_I \supseteq T$ and $\Delta_j(T) \leq \Delta_j(S_I)$.
	We have  know $\Delta_j(T)$ is smaller than the  value of $\Delta_{i}(S_{I})$ across nodes  $i \in S_{I+z}$, so
	\begin{align*}
	&Y \leq \Delta_j(T) \leq \Delta_{i\in S_{I+z}}(S_I) \leq \frac{1}{|S_{I+z}|} \sum_{i \in S_{I+z}} \Delta_i(S_I)
	\\
	&\leq \frac{1}{|S_{I+z}|}\left( \sum_{i \in S_{I+z}} d_i(S_I)^p + \sum_{i \in S_{I+z}} \sum_{l \in N(i) \cap S_I}( d_l(S)^p - [d_l(S_I)-1]^p) \right)
	\\
	&\leq \frac{ \sum_{i \in S_{I+z}} d_i(S_{I})^p}{|S_{I+z}|} + \frac{p \sum_{i \in S_{I+z}} d_i(S_{I})^{p}}{|S_{I+z}|}
		\\
	&\leq \frac{ \sum_{i \in S_{I}} d_i(S_I)^p}{|S_{I+z}|} + \frac{p \sum_{i \in S_{I}} d_i(S_I)^{p}}{|S_{I+z}|} \leq  2(p+1)f_p(S_I)
	\end{align*}
The last step above uses the fact that  $|S_{I+z}|\geq |S_I|/2 $, and third step above uses the observation 	$p(x-1)^{p-1} \leq x^p - (x-1)^p \leq px^{p-1}$
when  $p \geq 1$ and $x \geq 1$ which was obtained in by~\citet{veldt2021generalized}.
\end{proof}

\textbf{Remark}: Note that the $1/2$-factor in proof above can be replaced by any constant $0<c<1$. Then  the approximation ratio would be $((p+1)/(1-c))^{1/p}$.
	Since the current graph shrinks its size by $(1-c)$ factor in each iteration, Hence the total iteration is $\frac{\log n} {\log (1/(1-c))}$.
Subgraph ${S}$ satisfies that $\frac{p+1}{1-c} f_p({S})  \geq f_p(T)$, i.e., $(\frac{p+1}{1-c})^{1/p}M_p({S}) \geq  M_p(T)$ in $O(\frac{m\log n}{\log \frac{1}{1-c}})$ time for any constant $c$ $(0<c<1)$.
\vspace{0.15in}

\emph{Key differences between peeling algorithms.}
It is worth noting two key differences about removing one node each time based on $\Delta_j$ (\textsc{GENPEEL}) and removing half of the nodes based on $\Delta_j$ without update inside each deletion (\textsc{GenPeel++}). The latter one is purely practical: to keep track of changes to $\Delta_j$ for each node $j$ when peeling takes a lot of time. Since removing a node $v$ will not only change the degrees of nodes within a one-hop neighborhood, but will change $\Delta_j$ values for every node within a two-hop neighborhood of $v$. That is why GENPEEL takes $O(mn)$ time.
The analysis of time complexity of GENPEEL++ : There is only $O(\log n)
$ steps, the time complexity to compute/recompute $\Delta_j$ in each iteration is $O(m)$. Also, Sorting takes another $n\log n$ time in each iteration. However, the number of nodes of the graph is exponentially decreasing.
So the total time in sorting is $n\log n+n/2 (\log n/2)+n/4(\log n/4)+...+=O(n\log n)$. Hence the total time of Algorithm~\ref{alg:GENPEEL++} is $O((n+m)\log n)$.
In the next section, we will demonstrate empirically that \textsc{GenPeel++}  with $p > 1$  significantly outperforms \textsc{GENPEEL} in terms of the running time.

\section{Experiments}
\label{sec:experiments}

In this section, we present an experimental study on \textsc{GENPEEL++} for the task of locating dense subgraphs in real datasets in a fast and highly approximate manner. We will demonstrate the advantages  of \textsc{GENPEEL++} based on the performance and run-time speed. In terms of performance, we illustrate that \textsc{GENPEEL++} achieves a very close approximation to the optimal $p$-mean densest subgraph problem. To be specific, for $p \in [1, +\infty)$, it allows to discover subgraphs under varying meaningful notions of density. With regard to speed, a significant boost (e.g., by a factor of 10) is obtained by \textsc{GENPEEL++} over \textsc{GENPEEL}~\citep{veldt2021generalized}.

\newpage

\subsection{Implementation Details and Environment}

We conduct  our  experiments on a laptop with 8GB of RAM and 1.4 GHz Intel Core i5 processor, our experimental environment is close to the one in~\citet{veldt2021generalized} (where they used a laptop with 8GB of RAM and 2.2 GHz Intel Core i7). All methods were programmed in Julia. For \textsc{GENPEEL}, we remove nodes in min-heap data structure as in~\citet{veldt2021generalized}. As for \textsc{GENPEEL++}, we use the rather simpler array structure, which is more efficient for removing nodes. The optimal submodular minimization approach is implemented in MATLAB, in order to be handy for using existing submodular optimization package~\citep{krause2010sfo}. We shall mention that both \textsc{GENPEEL}~\citep{veldt2021generalized} and \textsc{GENPEEL++} used the same MATLAB package.

\subsection{Datasets}

We simply re-use the same  datasets as those used by~\citet{veldt2021generalized} including SNAP repository~\citep{snapnets} and the SuiteSparse Matrix collection~\citep{davis2011uf}. In order to compare \textsc{GENPEEL++} with \textsc{GENPEEL} for discovering dense subgraphs, we use benchmark graphs from different domains that are familiar to many. These datasets contain two citation networks (ca-Astro, condmat2005), two road networks (road-CA, road-TX), two web graphs (web-Google, web-BerkStan), an email network (Enron), two social networks (BrightKite, YouTube), and a retail graph (Amazon). To more closely demonstrate the change of average degree, size, maximum degree, and edge density changes as $p$, we run our experiments in all graphs from the Facebook100 dataset~\citep{traud2012social}.

\subsection{\textsc{GenPeel++} Approximation Performance}

In contrast to \textsc{GENPEEL}, considering the fact that \textsc{GENPEEL++} executes in the subgraph with a remaining node ratio of $1-c$, this may lead to a rather poor approximation. However, the dense subgraph found by \textsc{GENPEEL++}, is still ensured to be closer to the optimal solution than its worst case. Additionally, as described in~\citet{veldt2021generalized}, when using different $p$ by \textsc{GENPEEL++}, both the optimally-dense subgraphs for different $p$ and the sets found, are distinct from each other.

In order to verify the existence of the same occurrence in \textsc{GENPEEL++}, we take the experiments to find optimal solutions for more larger $p$ values in $P_\mathit{objs} = \{1.0, 1.5, 2.0, \ldots, 8.0\}$ to our objective as \textsc{GENPEEL}~\citep{veldt2021generalized}. Again,
the objective are solved via MATLAB implementation, which uses existing submodular minimization software on graphs with up to 1000 nodes to a small tolerance~\citep{krause2010sfo}. We then run \textsc{GenPeel++} for each $p \in P_\mathit{alg} =\{1.0,1.5,2.0,\dots,8.0\}$.
From the rounded curves in Figure~\ref{newfig:approx}, it is demonstrated that our problem is optimized differently according to each run of \textsc{GenPeel++}. For  Polbooks and Adjnoun, the approximation ratio is  better than $0.7$ for $p\in [1,8]$,  while the approximation ratio is  close to $0.9$ for Dolphins when $p\in [1,8]$.

Here we are also interested in comparing with \textsc{GENPEEL}~\citep{veldt2021generalized}. Similar to the experiment setting in~\citet{veldt2021generalized}, we experimented in Adjnoun and Jazz datasets. The subgraphs found with \textsc{GENPEEL++} are compared with those found with \textsc{GENPEEL}. Rounded curves are shown in Figure~\ref{fig:approx_ori}.
Note that the approximation ratio of \textsc{GENPEEL++} is  $(2(p+1))^{1/p}$, which converges to 1 when $p\rightarrow \infty$. Hence we are able to guarantee a high degree of approximation with \textsc{GENPEEL++}.
\begin{figure}[h]
\centering
\mbox{
	\subfloat[Polbooks, $n = 105$\label{newfig:polbooks}]
	{\includegraphics[width=3in]{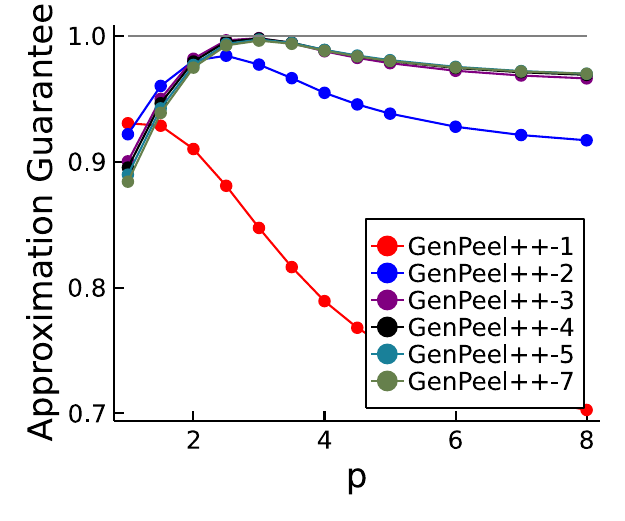}}\hspace{0.2in}
	\subfloat[Adjnoun, $n = 112$\label{newfig:adjnoun}]
	{\includegraphics[width=3in]{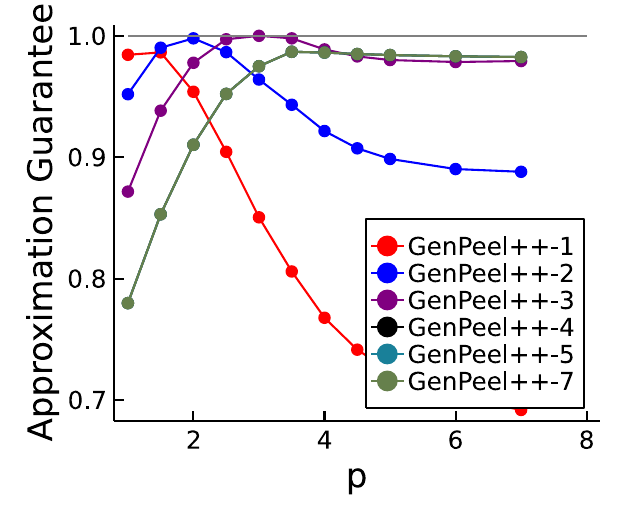}}
	}
	\mbox{
	\subfloat[Dolphins, $n = 198$ \label{newfig:jazz}]
	{\includegraphics[width=3in]{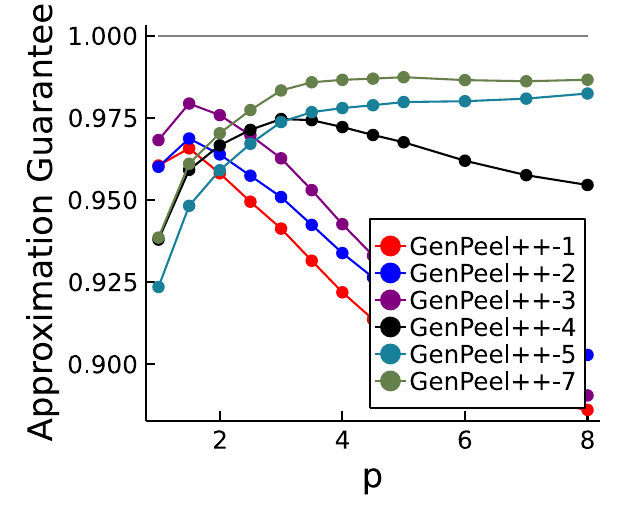}}\hspace{0.2in}
	\subfloat[Lesmis, $n = 1005$\label{newfig:email}]
	{\includegraphics[width=3in]{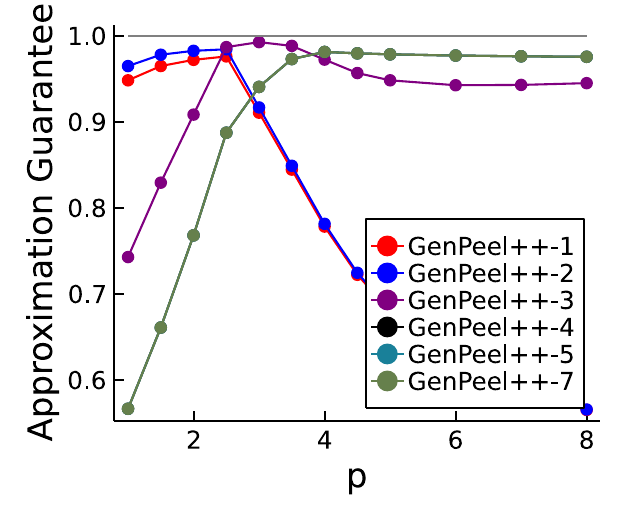}}
}

	\caption{Quality of the \textsc{GenPeel++} heuristic versus exact solution attained with submodular minimization.
	Different runs of \textsc{GenPeel++} make a good approximation to the objective corresponding to different values of $p$ from $1$ to $8$.
	}
	\label{newfig:approx}\vspace{-0.1in}
\end{figure}

\newpage\clearpage

\begin{figure}[h]
\centering
\mbox{
	\subfloat[Adjnoun, $n = 112$\label{fig:adjnoun}]
	{\includegraphics[width=3in]{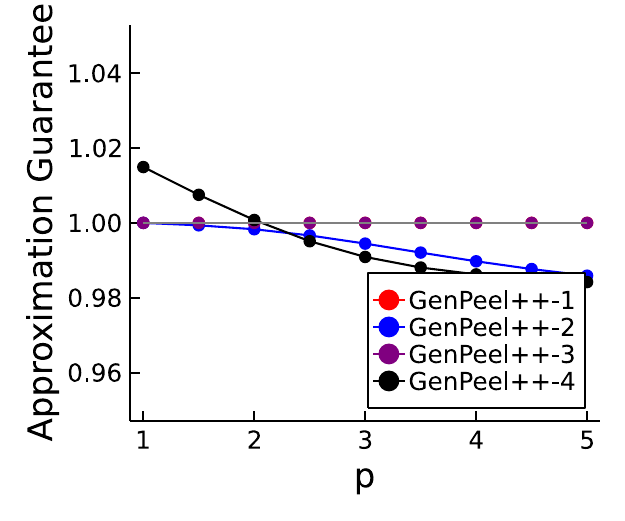}}\hspace{0.2in}
	\subfloat[Jazz, $n = 198$ \label{fig:jazz}]
	{\includegraphics[width=3in]{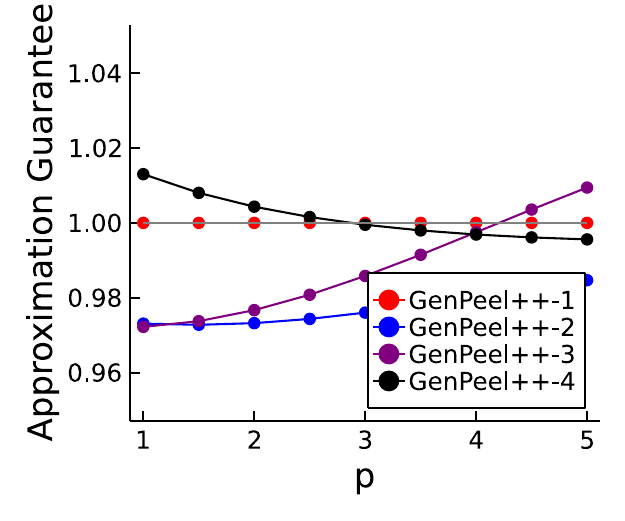}}
}
	\caption{Quality of the \textsc{GenPeel++} heuristic versus solutions obtained with \textsc{GENPEEL}~\citep{veldt2021generalized}, an approximation  provided by \textsc{GENPEEL++}.
	}
	\label{fig:approx_ori}\vspace{-0.2in}
\end{figure}

\subsection{Peeling Algorithms for Dense Subgraphs}

In the next set of experiments, we conduct experiments within the context of the broader parametric peeling algorithm for dense subgraph discovery, instead of the \textsc{SIMPEEL} ($p = 1$), and the peeling algorithm for finding maxcore ($p = -\infty$). With these outputs compared to different $p$ values of \textsc{GenPeel++} that are close to 1, an observation can be made that different datasets find the optimum dense subgraph in different $p$ values. In finding sets that satisfy the traditional notion of $p = 1$ density, we can also see how running \textsc{GenPeel++} for values close to, but not equal to 1 provides an accuracy versus runtime trade-off.  We choose $p\in\{0.5, 1.0, 1.05, 1.5, 2.0\}$. Note that \textsc{SIMPEEL}  has  constant approximation guarantees for $0<p<1$, but we found that  \textsc{GENPEEL++} outperforms than \textsc{SIMPLEEL} even for $0<p<1$ in those datasets. So we use  \textsc{GENPEEL++} for both case $p>1$ and case $0<p<1$. Besides, \textsc{GENPEEL++}
can greedily optimize the $p$-mean density target and produces meaningfully different subgraphs. We also focus on the comparison for $p \in [0.5,2]$ in order to make   the experimental setting here to be consistent with \textsc{GENPEEL}.

\vspace{0.1in}

In Table~\ref{tab:fast_snap}, we report the edge density (number of edges divided by number of pairs of nodes), the size of the set returned, and the average degree (i.e., the $p = 1$ objective). We can observe that subgraphs have a tendency to be smaller and have a higher density of edges as $p$ decreases. This observation is consistent with the previous one in \textsc{GenPeel}~\citep{veldt2021generalized}.
At the same time, there has been a significant increase in the average squared degree and the maximum degree. Possibly the highlight is that running \textsc{GenPeel++} with $p > 1$ prefer to yield a superior set than \textsc{SimplePeel} in the context of the standard densest subgraph objective. In terms of runtime, \textsc{GENPEEL++} is extremely efficient and works very rapidly even in very large graphs.
In Table~\ref{tab:approx_label}, we provide a demonstration of the approximate performance of \textsc{GenPeel++} against \textsc{GenPeel} in four graph datasets.  $\textbf{avg } d_v(S)^2$ and $\textbf{avg } d_v(S)^2$ are used as evaluation metrics. The experimental results show that  \textsc{GENPEEL++} is able to deliver a sound approximation to \textsc{GENPEEL}.

\begin{table*}[t!]

\vspace{0.3in}

	\caption{
	The parameter $c$ (ratio of preserved nodes) is set to $0.5$. We compare \textsc{GENPEEL++} for the $p$-mean densest subgraph with \textsc{SimplePeel} (the $p = 1$ special case) and the maxcore of a graph (p = $-\infty$) in different datasets across the different evaluation criteria, note that these datasets used are the same as ones in \textsc{GENPEEL}~\citep{veldt2021generalized}. \textsc{SimplePeel}  ($p = 1$) and maxcore cases take approximately the same amount of time to run, as they tend to rely on finding the same ordering of nodes with different stopping points.
Both \textsc{SimplePeel} and \textsc{GenPeel++} run fast. \textsc{GenPeel++} in several cases produce better results for edge density and average degree in comparison with \textsc{SIMPEEL}. We highlight in bold the result better than obtained for each kind of  density.
	}
	\label{tab:fast_snap}
	\centering
	\scalebox{0.68}{\begin{tabular}{l   |  l  | l l l     l l l  l l l l }
			\toprule
			&& \textbf{Astro} & \textbf{CM05} & \textbf{BrKite} & \textbf{Enron} & \textbf{roadCA} & \textbf{roadTX} & \textbf{webG} & \textbf{webBS} & \textbf{Amaz} & \textbf{YTube} \\
			&$|V|$& 17,903 &36,458 & 58,228 & 36,692 & 1,971,281 & 1,393,383 & 916,428 & 685,230 & 334,863 & 1,134,890  \\
			\textbf{Metric}& $|E|$ & 196,972 &171,734 &21,4078 &183,831 &2,766,607 &1,921,660 &4,322,051 &6,649,470 &925,872 &2,987,624 \\
			\midrule
			Size & \emph{maxcore}& 57& 30& 154& 275& 4568& 1579& 48& 392& 497& 845 \\
			& $p = 0.5$& 469  & 469  & 281  & 435  & 10  & 6  & 226  & 391  & 8  & 984 \\
			$|S|$& $p = 1.0$&1184  & 561  & 219  & 548  & 11  & 3721  & 240  & 392  & 34  & 1863\\
			& $p = 1.05$&1673  & 535  & 164  & 473  & 5  & 292  & 231  & 392  & 10  & 1793 \\
			& $p = 1.5$&890  & 821  & 182  & 850  & 5  & 231  & 221  & 5089  & 127  & 2985\\
			& $p = 2.0$& 915  & 964  & 201  & 1046  & 16  & 271  & 4450  & 34953  & 557  & 29725 \\
			\midrule
			Edge  & \emph{maxcore}& \textbf{1.0}& \textbf{1.0}& \textbf{0.502}& \textbf{0.256}& 0.001& 0.002& \textbf{0.994}& \textbf{0.529}& 0.014& \textbf{0.102} \\
			Density& $p = 0.5$& 0.123  & 0.064  & 0.255  & 0.165  & 0.333  & \textbf{0.600}  & 0.235  & \textbf{0.529}  & 0.964  & 0.085 \\
			& $p = 1.0$& 0.050  & 0.056  & 0.372  & 0.137  & 0.345  & 0.001  & 0.227  & \textbf{0.529}  & 0.230  & 0.049 \\
			$|E_S|/{|S| \choose 2}$& $p = 1.05$& 0.036  & 0.058  & 0.475  & 0.155  & \textbf{0.800}  & 0.012  & 0.228  & 0.529  & 0.867  & 0.049 \\
			& $p = 1.5$& 0.069  & 0.038  & 0.440  & 0.086  & \textbf{0.800}  & 0.015  & 0.240  & 0.033  & 0.073  & 0.029 \\
			& $p = 2.0$&0.068  & 0.032  & 0.404  & 0.067  & 0.225  & 0.013  & 0.005  & 0.001  & 0.005  & 06\\
			\midrule
			Avg  & \emph{maxcore}& 56.0& 29.0& 76.87& 70.06& 3.32& 3.34& 46.71& \textbf{206.81}& 6.77& 86.07 \\
			Degree& $p = 0.5$& 57.484  & 30.119  & 71.345  & 71.664  & 3  & 3  & 52.814  & 206.312  & 6.750  & 83.686 \\
			& $p = 1.0$& 59.231  & \textbf{31.554}  & \textbf{81.114}  & \textbf{74.682}  & \textbf{3.455}  & \textbf{3.491}  & \textbf{54.358}  & \textbf{206.811}  & 7.588  & \textbf{91.16} \\
			$\textbf{avg } d_v(S)$& $p = 1.05$& 59.423  & 30.819  & 77.476  & 73.353  & 3.200  & 3.459  & 52.442  & \textbf{206.811}  & 7.800  & 89.50 \\
			& $p = 1.5$& 61.218  & 30.916  & 79.637  & 72.616  & 3.200  & 3.455  & 52.769  & 167.787  & \textbf{9.228}  & 88.14 \\
			& $p = 2.0$& \textbf{61.864}  & 30.423  & 80.806  & 70.411  & 3.375  & 3.402  & 20.343  & 44.031  & 2.693  & 20.03 \\
			\midrule
			Avg & \emph{maxcore}& 3136.0& 841.0& 6335.5& 5685.5& 11.3& 11.7& 2182.4& 43840.3& 47.4& 9227.8 \\
			Squared & $p = 0.5$& 3688.064  & 1104.192  & 6048.164  & 6367.430  & 9  & 9  & 2862.389  & 43631.284  & 45.750  & 8893.49 \\
			Degree & $p = 1.0$& 4154.3& 1265.8& \textbf{7614.1} & 7301.6& 12.2& 12.7& 3031.9& 43840.3& 59.2& 12146.5 \\
			& $p = 1.05$& 4479.787  & 1181.929  & 6493.366  & 6786.359  & 10.400  & 12.753  & 2860.580  & 43840.281  & 61.600  & 11422.13 \\
			$\textbf{avg } d_v(S)^2$& $p = 1.5$& 4726.915  & 1298.037  & 6966.187  & 7530.878  & 10.400  & \textbf{12.788}  & 2843.955  & 150686.885  & 351.496  & 13504.90\\
			& $p = 2.0$& \textbf{5123.436}  & \textbf{1361.786}  & 7363.433  & \textbf{7855.231}  & \textbf{12.375}  & 12.613  & \textbf{9550.399}  & \textbf{455832.538}  & \textbf{543.820}  & \textbf{33238.71} \\
						\midrule
			Max & \emph{maxcore}& 56& 29& 153& 216& \textbf{7} & \textbf{12}& 47& 391& 13& 447 \\
			Degree & $p = 0.5$& 175  & 118  & 201  & 277  & 3  & 3  & 81  & 390  & 7  & 515 \\
			& $p = 1.0$& 272  & 163  & \textbf{214}  & 333  & 4  & \textbf{12}  & 84  & 391  & 10  & 954 \\
			$\textbf{max } d_v(S)$& $p = 1.05$& 299  & 136  & 160  & 299  & 4  & 8  & 82  & 391  & 9  & 890 \\
			& $p = 1.5$& 274  & 188  & 179  & 424  & 4  & 7  & 81  & 5088  & 106  & 1397 \\
			& $p = 2.0$& \textbf{312}  & \textbf{236}  & 197  & \textbf{502}  & 4  & 9  & \textbf{2295}  & \textbf{34942}  & \textbf{548}  & \textbf{28754} \\
						\midrule
			Runtime & \emph{maxcore}& 0.03& 0.03& 0.06& 0.04& 1.37& 0.96& 2.4& 11.27& 0.45& 2.51 \\
			& $p = 0.5$& 0.748  & 0.141  & 0.220  & 0.170  & 4.997  & 3.194  & 5.960  & 15.313  & 1.213  & 5.873\\
			& $p = 1.0$& 0.037  & 0.031  & 0.048  & 0.045  & 1.187  & 0.764  & 2.041  & 18.250  & 0.351  & 2.385 \\
			& $p = 1.05$& 0.140  & 0.206  & 0.218  & 0.170  & 5.516  & 3.785  & 6.060  & 15.628  & 1.218  & 5.85 \\
			& $p = 1.5$&0.137  & 0.144  & 0.219  & 0.166  & 5.300  & 3.636  & 5.952  & 15.550  & 1.255  & 6.08 \\
			& $p = 2.0$& 0.149  & 0.202  & 0.222  & 0.165  & 5.230  & 3.567  & 5.626  & 14.808  & 1.220  & 5.31 \\
			\bottomrule
	\end{tabular}}
\end{table*}

\newpage\clearpage

\begin{table*}[t!]
    \caption{Comparison of \textsc{GENPEEL++} and \textsc{GENPEEL} in four datasets: BrKite, roadCA, webBS, YTube. $\textbf{avg } d_v(S)$ and $\textbf{avg } d_v(S)^2$ are used as evaluation metrics. The experimental results show that  \textsc{GENPEEL++} provides a sound approximation to \textsc{GENPEEL}. The light grey numbers in brackets indicate the value of the approximate ratio of \textsc{GENPEEL} to \textsc{GENPEEL}.}
    \centering
    \begin{tabular}{l | l | c  c | c  c}
    \toprule
    \multicolumn{2}{l}{Method} & \multicolumn{2}{c|}{\textsc{GENPEEL++}} & \multicolumn{2}{c}{\textsc{GENPEEL~\citep{veldt2021generalized}}} \\
    \midrule
    Dataset &   &$\textbf{avg } d_v(S)$ & $\textbf{avg } d_v(S)^2$ & $\textbf{avg } d_v(S)$ & $\textbf{avg } d_v(S)^2$ \\
    \hline
    \multirow{4}{*}{BrKite} & $p = 0.5$  & 71.345 (\textcolor{shadecolor}{0.88}) & 6048.164 (\textcolor{shadecolor}{0.82}) & 80.91  & 7372.9 \\
                            & $p = 1.05$ & 77.476 (\textcolor{shadecolor}{0.96}) & 6493.366 (\textcolor{shadecolor}{0.86}) & 81.12 & 7624.9 \\
                            & $p = 1.5$  & 79.637 (\textcolor{shadecolor}{0.99}) & 6966.187 (\textcolor{shadecolor}{0.90}) & 80.8 & 7776.6 \\
                            & $p = 2.0$  & 80.806 (\textcolor{shadecolor}{1.02})& 7363.433 (\textcolor{shadecolor}{0.9}) & 78.99 & 7882.1 \\
    \hline
    \multirow{4}{*}{roadCA} & $p = 0.5$  & 3 (\textcolor{shadecolor}{0.90})   & 9 (\textcolor{shadecolor}{0.80})    & 3.32 & 11.3 \\
                            & $p = 1.05$ & 3.200 (\textcolor{shadecolor}{0.87}) & 10.400 (\textcolor{shadecolor}{0.76}) & 3.67 & 13.7 \\
                            & $p = 1.5$  & 3.200 \textcolor{shadecolor}{0.87}& 10.400 \textcolor{shadecolor}{0.76} & 3.67 & 13.7 \\
                            & $p = 2.0$  & 3.375 (\textcolor{shadecolor}{0.93}) & 12.375 (\textcolor{shadecolor}{0.89})  & 3.62 & 13.9 \\
    \hline
    \multirow{4}{*}{webBS} & $p = 0.5$   & 206.312 (\textcolor{shadecolor}{0.99}) & 43631.284 (\textcolor{shadecolor}{0.99}) & 206.81 & 43840.3 \\
                            & $p = 1.05$ & 206.811 (\textcolor{shadecolor}{1.00}) & 43840.281 (\textcolor{shadecolor}{1.00}) & 206.81 & 43840.3 \\
                            & $p = 1.5$  & 167.787 (\textcolor{shadecolor}{1.01} & 150686.885 \textcolor{shadecolor}{0.99} & 166.93 & 157225.3 \\
                            & $p = 2.0$  & 44.031 (\textcolor{shadecolor}{1.00}) & 455832.538 (\textcolor{shadecolor}{1.00})& 44.04 & 455975.9 \\
    \hline
    \multirow{4}{*}{YTube} & $p = 0.5$   & 83.686 (\textcolor{shadecolor}{0.92}) & 8893.49 (\textcolor{shadecolor}{0.77}) & 90.76 & 11486.8 \\
                            & $p = 1.05$ & 89.50 (\textcolor{shadecolor}{0.98}) & 11422.13 (\textcolor{shadecolor}{0.93}) & 91.18 & 12220.1 \\
                            & $p = 1.5$  & 88.14 (\textcolor{shadecolor}{0.99}) & 13504.90 (\textcolor{shadecolor}{0.94}) & 88.86 & 14359.8 \\
                            & $p = 2.0$  & 20.03 (\textcolor{shadecolor}{1.01})& 33238.71 (\textcolor{shadecolor}{1.00}) & 19.88 & 33262.3 \\
    \bottomrule
    \end{tabular}
    \label{tab:approx_label}
\end{table*}

\newpage\clearpage

\begin{figure}[h]
	\centering
	\subfloat[Edge Density in \textsc{GL++} \label{fig:edge_gen++}]
	{\includegraphics[width=.25\linewidth]{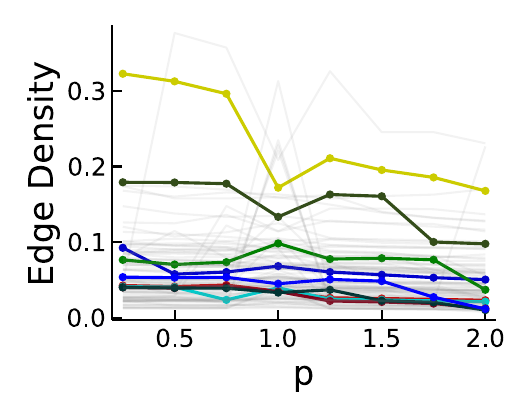}}\hfill
	\hspace{-0.1in}
	 \subfloat[Edge Density in \textsc{GL} \label{fig:edge_gen}]
	{\includegraphics[width=.25\linewidth]{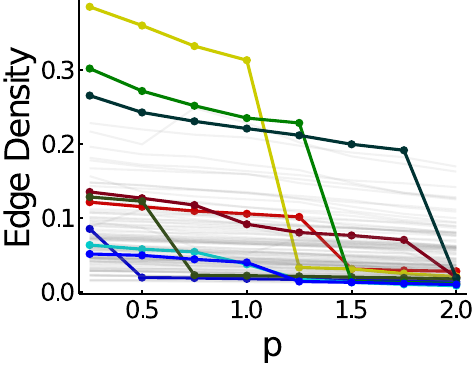}}\hfill
\hspace{-0.1in}
	\subfloat[Size in \textsc{GL++} \label{fig:size_gen++}]
	{\includegraphics[width=0.25 \linewidth]{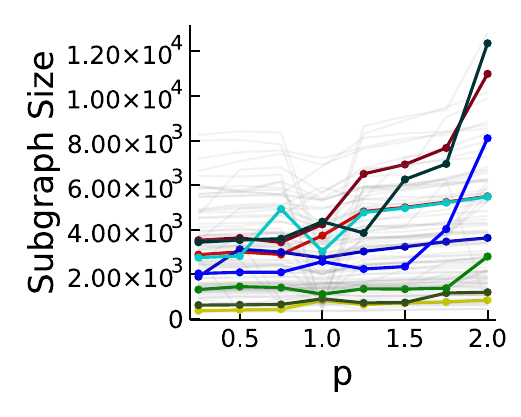}}\hfill
	\subfloat[Size in \textsc{GL}  \label{fig:size_gen} ]
	{\includegraphics[width=0.25 \linewidth]{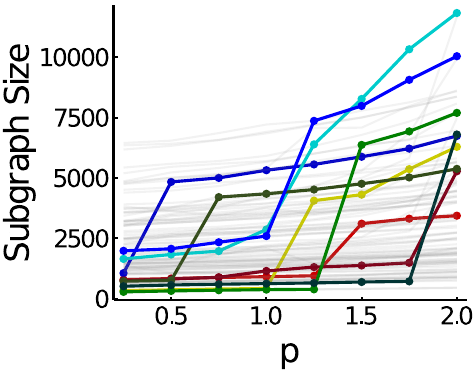}}

 \vspace{-0.1in}

 \subfloat[Max Degree in \textsc{GL++} \label{fig:max_degree_gen++}]
	{\includegraphics[width=.25\linewidth]{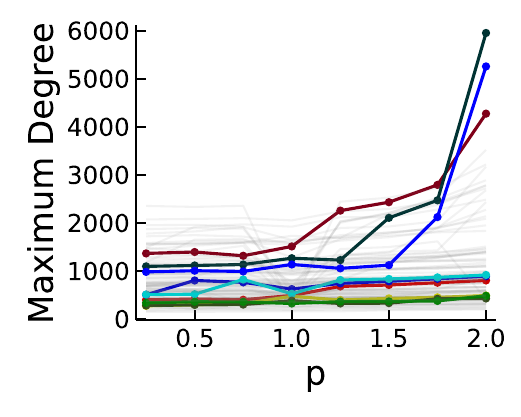}}
	\subfloat[Max Degree in \textsc{GL}  \label{fig:max_degree_gen} ]
	{\includegraphics[width=.25\linewidth]{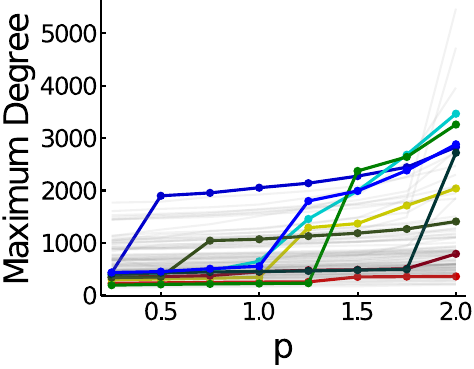}}
	\subfloat[Mean Degree in \textsc{GL++} \label{fig:mean_degree_gen++}]
	{\includegraphics[width=.25\linewidth]{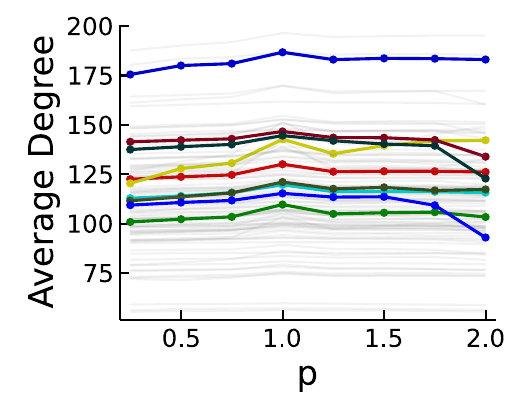}}\hfill
	\subfloat[Mean Degree in \textsc{GL} \label{fig:mean_degree_gen} ]
	{\includegraphics[width=.25\linewidth]{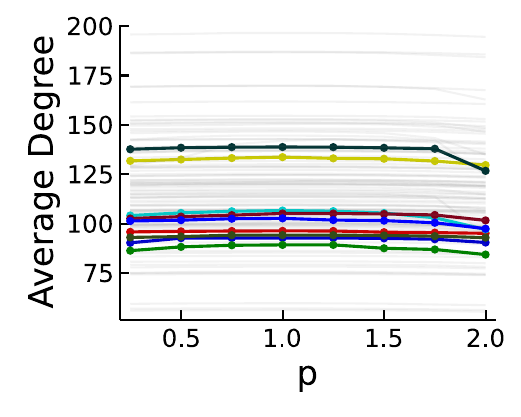}}\hfill

	\caption{Comparison between \textsc{GENPEEL++} (``GL++'' for short) and \textsc{GENPEEL} (``GL'' for short)~\citep{veldt2021generalized} in Facebook for $p \in [0.25,2]$.  Following~\citet{veldt2021generalized},  each line (total 100 lines ) represents a graph. There is a lot of overlap in the set. The 10 networks with the top most significant size changes are in color  to show their leading trends. (a) Edge density in \textsc{GENPEEL++}. (b)  Edge density in \textsc{GENPEEL}. (c) Size in \textsc{GENPEEL++}.  (d) Size in GENPEEL. (e) Maximum Degree in \textsc{GENPEEL++}. (f) Maximum Degree in \textsc{GENPEEL}. (g) Mean Degree in \textsc{GENPEEL++}. (h) Mean Degree in \textsc{GENPEEL}. }
	\label{fig:fb}
\end{figure}

\begin{figure}[b!]
\vspace{-0.28in}
\centering
\mbox{
	\includegraphics[width=2.9in]{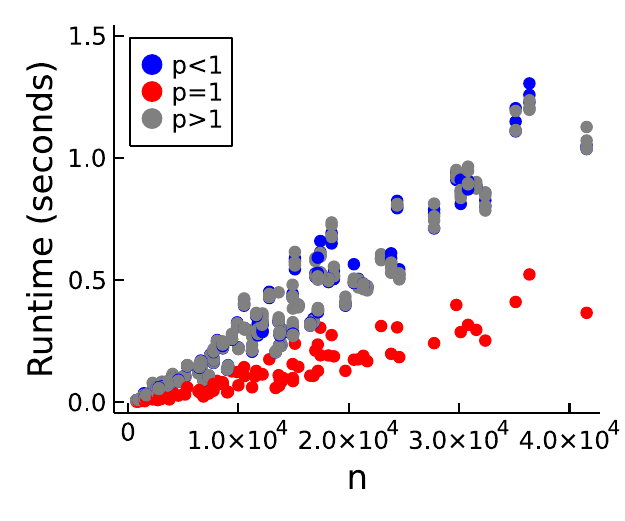}
}

\vspace{-0.25in}

	\caption{Runtime for Facebook100. \textsc{GenPeel++} takes roughly 1.5 seconds~to~run in worst case.}
	\label{fig:fbrun-fast}\vspace{-0.3in}
\end{figure}

\textbf{Dense Subgraphs in Social Networks.}
Then we positioned diverse dense subgraphs across the Facebook100 dataset using \textsc{GenPeel++}. Each graph Among the Facebook100 dataset is a snapshot of the American University Facebook network. We run \textsc{GENPEEL++} for $p\in\{0.5, 1.0, 1.05, 1.5, 2.0\}$
In Figure~\ref{fig:fb}, we plot the variation in  edge density, size, maximum degree, and mean degree varying $p$ for each graph. As $p$ increases, the maximum degree of the subgraph follows (Figure~\ref{fig:max_degree_gen++}), Conversely, as $p$ decreases, the resulting set will have a higher edge density (Figure~\ref{fig:edge_gen++}).
The mean degree in (Figure~\ref{fig:mean_degree_gen++}) almost does not change too much for different $p$ values.
These observations are consistent with those in~\textsc{GenPeel}. Mindful the difference here is that~\textsc{GenPeel++} produces significantly higher edge density than~\textsc{GenPeel} when $p$ is large. We also measured the runtime of the algorithm for different subgraphs in Facebook100 dataset. Figure~\ref{fig:fbrun-fast} is a scatter plot of points ($n$,$s$) where $n$ is node number in a Facebook graph and $s$ the running time in seconds. When $p = 1$, we use the \textsc{SIMPEEL} algorithm which is consistent with \textsc{GENPEEL}. For $p\neq 1$, in this case, \textsc{GENPEEL} takes a little more time. However, overall, \textsc{GenPeel++} still requires very little time and runs very fast.
In order to be able to assess the benefit of the new algorithm \textsc{GenPeel++}, more study for the full range of $p$ is given in subsection below.

\subsection{Speedup in Larger Graph Datasets}
\label{sec:more_experiment}
We had illustrated the advantages of \textsc{GENPEEL++} over \textsc{GenPeel}, on larger graphs in Table~\ref{tab:fast_snap}. Focusing specifically in the webBS and webG datasets, we can demonstrate the benefits of \textsc{GENPEEL++} in terms of speedup. From Figure~\ref{fig:p}, we can see that, subject to a small gap in Avg Degree, \textsc{GENPEEL++} is able to achieve up to a \textbf{40x} speedup in WebBS, and an \textbf{8x} speedup in webG. This confirms  the speed of \textsc{GENPEEL++} in large graphs is remarkable.

\begin{figure}[h]

\mbox{
\includegraphics[width=0.55\linewidth]{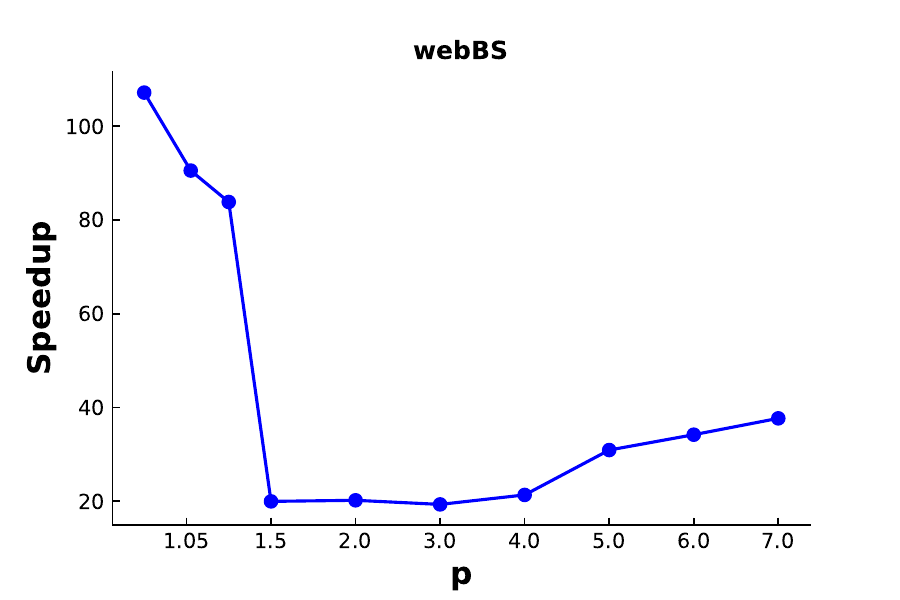}\hspace{-0.3in}	\includegraphics[width=0.55\linewidth]{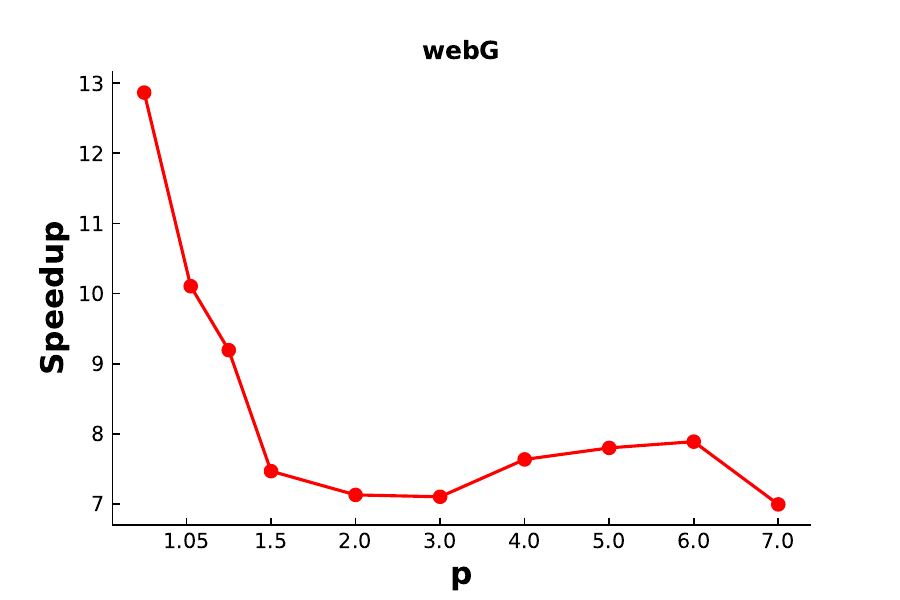}
}

	\caption{The speedup ratios of \textsc{GENPEEL++} versus \textsc{GENPEEL} in webBS (left) and webG (right) graph datasets.
	}
	\label{fig:p}\vspace{-0.15in}
\end{figure}

\section{Conclusion}
The $p$-mean densest subgraph objective  provides a general framework for capturing different notions of density in the same graph. In this paper, we have designed a faster approximation algorithm (GENPEEL++) with $O(m\log n)$ time, which improves the previous $O(mn)$ time algorithm. In GENPEEL++, we show that it is not necessary to update the ``insight'' information (the impact on the objective when a node is removed) once a node is removed.
Also, for $p \in (0, 1)$, we are able to show that the standard peeling algorithm yields a constant approximation. This means a single peeling algorithm on the nodes can be used to define a nested set of dense subgraphs that can well approximate our objective for a wide range of $p$ values. Our extensive
experimental results  have shown that the proposed GENPELL++ algorithm achieves really close approximations compared to the previous GENPEEL algorithm,
	and GENPELL++  performs significantly faster than  GENPEEL in large real datasets, coming from numerous domains.

 \newpage

\bibliography{refs_scholar}
\bibliographystyle{plainnat}

\end{document}